\documentclass[letterpaper, 10 pt, conference]{ieeeconf}  

\IEEEoverridecommandlockouts                              
\usepackage{cite}
\usepackage{amsmath,amssymb,amsfonts}
\usepackage{algorithmic}
\usepackage{graphicx}
\usepackage{textcomp}
\usepackage{xcolor}
\usepackage[scr]{rsfso}
\usepackage{soul}
\usepackage{hyperref}
\usepackage{url}
\usepackage{float}
\usepackage{titlesec}

\newtheorem{theorem}{Theorem}

\newtheorem{assumption}{Assumption}

        \allowdisplaybreaks
\usepackage{xcolor}

\begin{document}
\title{\bf Moving-Horizon Estimators\\ for Hyperbolic and Parabolic  PDEs in 1-D}

\author{Luke Bhan,
        Yuanyuan Shi, 
        Iasson Karafyllis,
         Miroslav Krstic, 
         and James B. Rawlings}


\maketitle


\begin{abstract}
Observers for PDEs are themselves PDEs. Therefore, producing real time estimates with such observers is computationally burdensome. For both finite-dimensional and ODE systems, moving-horizon estimators (MHE) are operators whose output is the state estimate, while their inputs are the initial state estimate at the beginning of the horizon as well as the measured output and input signals over the moving time horizon. In this paper we introduce MHEs for PDEs which remove the need for a numerical solution of an observer PDE in real time. We accomplish this using the PDE backstepping method which, for certain classes of both hyperbolic and parabolic PDEs, produces moving-horizon state estimates explicitly. Precisely, to explicitly produce the state estimates, 
we employ a backstepping transformation of a hard-to-solve observer PDE into a target observer PDE, which is explicitly solvable. 
The MHEs we propose are not new observer designs but simply the explicit MHE realizations, over a moving horizon of arbitrary length, of the existing backstepping observers. Our PDE MHEs lack the optimality of the MHEs that arose as duals of MPC, but they are given explicitly, even for PDEs. In the paper we provide explicit formulae for MHEs for both hyperbolic and parabolic PDEs, as well as simulation results that illustrate theoretically guaranteed convergence of the MHEs. 
\end{abstract}

\section{Introduction}

\paragraph{Motivation: moving-horizon estimators for PDEs}    In this paper, we develop moving-horizon estimators for \textit{two} distinct classes of PDEs: first-order hyperbolic partial integro-differential equation (PIDE) systems and parabolic PDE systems governing reaction-diffusion processes. Our motivation stems from a well-established challenge in the field that PDE observers are expensive to compute in real time due to being PDEs themselves. Instead, we craft explicit representations of the PDE observer systems, bypassing the challenges of PDE discretization schemes like finite differences, which, when mishandled, can lead to numerical instability as in the Gibbs phenomenon. To achieve explicit representations of PDE observers, we employ the power of PDE backstepping transformations. By designing these transformations, we convert complex observer PDE systems, which cannot be analytically solved, into  target systems that are analytically solvable. We then go a step further and, rather then producing state estimates using the measured signal history from the initial time, we employ the observer PDE's explicit solution over a moving time horizon, requiring only a measurement history over a constant, most recent time interval. This represents the notion of a \textit{moving-horizon estimator} (MHE), extended from \cite{Rawlings2003} to PDE systems. The MHEs we produce are the equivalent representations of the PDE observers and maintain their guarantees of exponential convergence towards the state of the plant. Thus, in this paper, we present a blueprint for a new direction in PDE observer design---PDE MHEs---on two separate PDE benchmark systems from distinct PDE classes.  

    \paragraph{Advances in PDE observers}
        The first results in PDE observers began with \cite{SMYSHLYAEV2005613} presenting a series of observers for 1-D parabolic PDEs using backstepping transformations. Following this result, observers were developed across a variety of PDE classes and even coupled PDE-ODE systems. For example, for hyperbolic PDEs, references \cite{6160338, 6573344, BERNARD20142692} present results for both complex coupled hyperbolic PDE-ODE systems as well as adaptive observer designs. Reference \cite{LIU201661} develops observers for coupled parabolic PDEs with different diffusion coefficients which represent simplified chemical tubular reactor models. 

        Additionally, observers have been developed for targeted applications such as the Burgers equation and more challenging Navier-Stokes equations in \cite{10.1115/1.3023128, 4177273, 1583115, VAZQUEZ20082517}. State-of-charge (SOC) estimation for electrochemical models is studied in \cite{socMoura} and \cite{5208259} studies a boundary observer for 2-D thermal fluid convection. Lastly,  a prescribed-time observer for the linear Schr\"{o}dinger equation is presented in \cite{STEEVES20203}.
        

    \paragraph{Moving horizon estimators}
    For many years, MHEs have been well studied for finite dimensional systems \cite{Muske1995}, \cite{RAO20011619}, \cite{4792993},  \cite[Ch. 4]{rawlings:mayne:diehl:2020},  \cite{Rawlings2003}, \cite{Rawlings20061529}, \cite{10.1007/978-1-4471-0853-5_19}, \cite{https://doi.org/10.1002/aic.690420811}. Particularly,
    MHEs are robustly globally exponentially stable (RGES) under mild exponential assumptions \cite[S. 4.3.3]{rawlings:mayne:diehl:2020}. Additionally, reference \cite{1039801} explores and provides conditions for convergence of MHEs in hybrid systems. From an application perspective, MHEs have found success over traditional approaches such as the extended Kalman filter in networked chemical reaction systems \cite{doi:10.1021/ie034308l} \cite{KUNTZ2023108429}. Moreover, reference \cite{8685160} develops a MHE algorithm for distributed networking systems and lastly, MHEs have played an important role in developing robots for various applications ranging from underwater localization to agricultural machinery \cite{WANG20141581} \cite {7572119 }\cite{KRAUS201325}.

    However, to the author's knowledge, very little work has been explored in MHEs for PDEs without first discretizing the PDE (early lumping) and then performing well known finite dimensional MHE \cite{JANG2014159}. Recently, reference \cite{DONGMO202285} proposed a late lumping approach---but only for reaction diffusion PDEs and builds their method from gradient-based estimation rather than explicitly solving the observer PDE as done in this work.

    \paragraph{Contributions}
    
The paper's contributions are for two separate PDE classes. First, we present a novel hyperbolic PIDE observer. We then prove that the observer exponentially converges to the true system state via  Lyapunov analysis. We then convert the PDE  observer into an explicit MHE, retaining \textit{all} the convergence guarantees of the PDE  observer. We repeat the process for parabolic PDEs, leveraging some of our results from \cite{SMYSHLYAEV2005613}. Lastly, in the case of parabolic PDEs we conclude with illustrations of our result via a simulation of the explicit MHE for a slowly unstable yet sufficiently excited parabolic PDE.

\paragraph{Notation}
$\mathbb{R}^+$ is the interval $[0, +\infty)$. Let $ u: [0,1] \times \mathbb{R}^+  \rightarrow \mathbb{R} $ be given. We use $u[t]$ to denote the profile of $u$ at $t$, i.e., $(u[t])(x)=u(x,t)$ for all $x \in [0,1]$. $L^2(0, 1)$ denotes the space of equivalent classes of square-integrable, measurable functions defined on $(0, 1)$. For every $u$ in $L^2(0, 1)$, $\|u\|$ denotes the standard norm of $u$ in $L^2(0, 1)$. $H^2(0, 1)$ is the Sobolev space of functions in $L^2(0, 1)$ with weak first and second derivatives in $L^2(0, 1)$. Let $u_x(x, t)$ denote the partial derivative of the function $u$ with respect to spatial variable $x$ and $u_t(x, t)$ denote similarly for $t$. 
\section{1D Hyperbolic PIDE} \label{sec:hyperbolic}
We consider hyperbolic PIDE systems commonly found in explaining physical phenomena such as traffic flows and chemical reactors \cite{krstic2008boundary} as
\begin{eqnarray}
    \label{eq:hyperbolicSys1}
    u_t(x, t) \nonumber &=& u_x(x, t) + g(x) u(0, t) \\ && + \int_0^x f(x, y)u(y, t) dy\,, \\ 
    u(1, t) &=&  U(t) \,, \label{eq:hyperbolicSys2}\\ \label{eq:hyperbolicSys3}
    Y(t) &=& u(0, t)\,,
\end{eqnarray}
where $t \geq 0$ is the time, $x \in [0, 1]$ is the spatial coordinate, $U$ is the control and $Y$ is the measured output. Furthermore, we require the following assumption.  
\begin{assumption}\label{assmp1}
    $f \in C^1(\mathcal{T})$ and $g \in C^1([0, 1])$, where $\mathcal{T} = \{(x, y) \in \mathbb{R}^2 ; 0 \leq y \leq x \leq 1\}$.
\end{assumption}

For every $u_0 \in C([0, 1])$ and $U \in C^1(\mathbb{R}^+)$ that satisfy the compatibility conditions
\begin{eqnarray}
    u_0(1) &=& U(0) \label{eq:hyperbolicCompat1} \,, \\ 
    \dot{U}(0) &=& u_0'(1) + g(1)u_0(0) + \int_0^1 f(1, y) u_0(y) dy \,, \label{eq:hyperbolicCompat2} 
\end{eqnarray}
using Corollary 2.3
in \cite{iasson}, we guarantee that the PDE system \eqref{eq:hyperbolicSys1}, \eqref{eq:hyperbolicSys2} with initial condition
\begin{eqnarray}
    u[0]=u_0 \label{iasson1} \ 
\end{eqnarray}
has a unique classical solution defined for all $t \geq 0$.

Moreover, in relation to Assumption \ref{assmp1}, we define
    \begin{eqnarray}
        M_g &=& \sup_{x \in [0, 1]}|g(x)|  \label{eq:hyperbolicAssum1}\,,\\
        M_f &=& \sup_{(x, y) \in \mathcal{T}}|f(x, y)| \,. \label{eq:hyperbolicAssum2}
    \end{eqnarray}

Next, we postulate the  observer
\begin{eqnarray}
    \hat{u}_t(x, t) &=& \hat{u}_x(x, t) + g(x)\hat{u}(0, t) \nonumber \\ && + \int_0^xf(x, y)\hat{u}(y, t)dy \nonumber \\ && + p_1(x)[u(0, t)-\hat{u}(0, t)] \,,\label{eq:hyperbolicObs1}\\ 
    \hat{u}(1, t) &=& U(t)\,,\label{eq:hyperbolicObs2}
\end{eqnarray}
%
%
with a (Dirichlet) boundary measurement \eqref{eq:hyperbolicSys3}
at $x=0$, and with the observer gain function $p_1(x)$ to be derived. As above, we require that $\hat{u}_0 \in C^1([0, 1])$ and that the following compatibility conditions
\begin{eqnarray}
    \hat{u}_0(1) &=& U(0) \,, \label{eq:hyperbolicObsCompat1}\\ 
    \dot{U}(0) &=& \hat{u}'_0(1) + g(1)\hat{u}_0(0) \nonumber \\ && \nonumber + \int_0^1 f(1, y) \hat{u}_0(y) dy \\ 
    && + p_1(1) \big[ u_0(0) - \hat{u}_0(0) \big] \label{eq:hyperbolicObsCompat2} \,,
\end{eqnarray}
are satisfied to ensure the system \eqref{eq:hyperbolicSys1}, \eqref{eq:hyperbolicSys2}, \eqref{eq:hyperbolicObs1}, \eqref{eq:hyperbolicObs2} with initial conditions \eqref{iasson1} and 
\begin{eqnarray}
    \hat{u}[0]=\hat{u}_0 \label{iasson2} \ 
\end{eqnarray}
has a unique classical solution defined for all $t \geq 0$. 

Further, let the state estimation error be denoted by 
\begin{eqnarray}
    \tilde{u}(x, t) = u(x, t) - \hat{u}(x, t)\,,
\end{eqnarray} 
resulting in the error system 
    \begin{eqnarray}
        \label{eq:hyperbolicErr1} \tilde{u}_t(x, t) &=& \tilde{u}_x(x, t) + g(x)\tilde{u}(0, t) \nonumber \\ &&+ \int_0^x f(x, y) \tilde{u}(y, t) dy \nonumber \\ && - p_1(x)\tilde{u}(0, t) \,,\\ \label{eq:hyperbolicErr2}
        \tilde{u}(1, t) &=& 0 \,.
    \end{eqnarray}
    Now, consider applying the backstepping transformation
    \begin{eqnarray}
        \tilde{u}(x, t) = \tilde{w}(x, t) - \int_0^x k(x, y) \tilde{w}(y, t) dy \,, \label{eq:hyperbolicBcskTrnsfm}
    \end{eqnarray}
    with the goal of ensuring that the state $\tilde{w}$ of the observer error target system be governed by the plain homogeneous transport PDE
    \begin{eqnarray}
        \tilde{w}_t(x, t) &=& \tilde{w}_x(x, t) \label{eq:hyperbolicErrTarget1} \,,\\ 
        \tilde{w}(1, t) &=& 0 \label{eq:hyperbolicErrTarget2}\,.
    \end{eqnarray}
    To transform \eqref{eq:hyperbolicErr1}, \eqref{eq:hyperbolicErr2} into this target system, $k(x, y)$ must  satisfiy the Goursat PDE
\begin{eqnarray}
    \label{eq:hyperbolicKernelCond1} k_y(x, y) + k_x(x, y) &=& f(x, y) \nonumber \\ && - \int_y^x  f(x, \eta) k(\eta, y) d \eta \nonumber\,,  \\ &&  \forall (x, y) \in \mathcal{T} \,, \\ 
    k(1, y) &=& 0, \hspace{5pt} \forall y \in [0, 1] \,,\label{eq:hyperbolicKernelCond2} 
\end{eqnarray}
and the observer gain $p_1(x)$ must be given by
\begin{eqnarray}
    p_1(x) = g(x) - k(x, 0)\,. \label{eq:hyperbolicObsGain}
\end{eqnarray}
The following result certifies that the kernel PDE $k(x, y)$ is bounded.

\begin{theorem} \label{thm:kernelBound}
    (proven in \cite{BERNARD20142692}) For every $f \in C^1(\mathcal{T})$ and $g \in C^1([0, 1])$ the PDE problem \eqref{eq:hyperbolicKernelCond1}, \eqref{eq:hyperbolicKernelCond2}, has a unique $C^1(\mathcal{T})$ solution with the bound 
    \begin{eqnarray}
        |k(x, y)| \leq M_f(1-x) e^{M_f(x-y)(1-x)} \,,\hspace{10pt} \forall (x, y) \in \mathcal{T}\,.\label{eq:hyperbolicKernelBound}
    \end{eqnarray}
\end{theorem}
\vspace{5pt} 

 The result is proven in Sec 6.1 of \cite{BERNARD20142692} using the method of successive approximations coupled with a change of variables. Following Theorem \eqref{thm:kernelBound}, we define
 \begin{eqnarray}
     \bar{k} = \sup_{(x, y) \in \mathcal{T}}|k(x, y)|
 \end{eqnarray}

Furthermore, recall the inverse backstepping kernel $l(x, y) \in C^1(\mathcal{T})$ and the inverse backstepping transformation as
\begin{eqnarray}
    \tilde{w}(x, t) = \tilde{u}(x, t) + \int_0^x l(x, y) \tilde{u}(y, t) dy \,. 
\end{eqnarray}
From Section 4.5 of \cite{krstic2008boundary}, the backstepping and inverse backstepping kernels satisfy 
\begin{eqnarray}
    l(x, y) &=& k(x, y) + \int_y^x k(x, \eta) l(\eta, y)  d y\,. \label{eq:hyperbolicKernelProp}
\end{eqnarray}
leading to the following bound on the inverse kernel
\begin{eqnarray}
    |l(x, y)| \leq \bar{k}  e^{\bar{k}}\,.
    \label{eq:hyperbolicInvBound}
\end{eqnarray}
Then, applying \eqref{eq:hyperbolicKernelBound} to $k(x, y)$ in \eqref{eq:hyperbolicInvBound} yields the following conservative bound
\begin{eqnarray}
    |l(x, y)| \leq \left(M_fe^{M_f} \right) e^{M_fe^{M_f}} \label{eq:hyperbolicInvBound1}\,.
\end{eqnarray} Similar to above, we then define
\begin{eqnarray}
    \bar{l} = \sup_\mathcal{T}|l(x, y)|
\end{eqnarray}

We now have the tools to present the first result for the observer in \eqref{eq:hyperbolicObs1}, \eqref{eq:hyperbolicObs2}.

\begin{theorem} \label{thm:hyperbolicErrConverges}
    Let Assumption \ref{assmp1} hold. Then for every $U \in C^1(\mathbb{R}^+)$, $u_0, \hat{u}_0 \in C^1([0, 1])$ for which the compatibility conditions in \eqref{eq:hyperbolicCompat1}, \eqref{eq:hyperbolicCompat2},
    \eqref{eq:hyperbolicObsCompat1}, \eqref{eq:hyperbolicObsCompat2} hold, the observer \eqref{eq:hyperbolicObs1}, \eqref{eq:hyperbolicObs2} 
ensures that the error between its state with initial condition \eqref{iasson2} and the state of the plant \eqref{eq:hyperbolicSys1}, \eqref{eq:hyperbolicSys2}, \eqref{eq:hyperbolicSys3} with initial condition \eqref{iasson1} satisfies for all $t \geq 0$ and $c>0$ the exponential stability bound
\begin{eqnarray}
    \|u[t] - \hat{u}[t]\| \leq Me^{-c t}\|u_0-\hat{u}_0\|\,,
\end{eqnarray}
where 
\begin{eqnarray}
    M = e^{c} (1 + \left(M_fe^{M_f} \right) e^{M_fe^{M_f}})(1 + M_fe^{M_f})\,,
\end{eqnarray}
Moreover, 
$\hat u(t) \equiv {u}(t)$ for $t\geq 1$, namely, the observer error is equal to $0$ for $t \geq 1$ and remains zero thereafter. 
\end{theorem}

\begin{proof}
    First, note that the observer can be transformed into the target system $\tilde{w}$ as in \eqref{eq:hyperbolicErrTarget1}, \eqref{eq:hyperbolicErrTarget2} using the backstepping transformation in \eqref{eq:hyperbolicBcskTrnsfm}. We propose the Lyapunov functional for the $\tilde{w}$ system as
    \begin{eqnarray}
        V(t) = \frac{1}{2} \int_0^1 e^{cx} \tilde{w}^2 (x, t) dx, \qquad c> 0\,,
    \end{eqnarray}
    with derivative
    \begin{eqnarray}
        \dot{V} &=& -\tilde{w}^2(0, t) - c \int_0^1 e^{cx} \tilde{w}^2 (x, t) dx \\ &=& -\tilde{w}^2(0, t) - 2cV \,.
    \end{eqnarray}
    Note that $V$ satisfies 
    \begin{eqnarray}
        \frac{1}{\left(1+\bar{k}\right)^2}\|\tilde{u}\|^2 \leq V \leq e^{c} \left(1+\bar{l} \right)^2 \|\tilde{u}\|^2\,, \label{eq:hyperbolicVSandwich} 
    \end{eqnarray}
    and since $V(t) \leq V(0)e^{-2ct}$, plugging in for \eqref{eq:hyperbolicVSandwich} yields
    \begin{eqnarray}
        \|\tilde{u}[t]\| \leq e^{c(1-t)}(1+\bar{l})(1+\bar{k}) \|\tilde{u}[0]\|\,.
    \end{eqnarray}

    Lastly, note that by explicit calculation of \eqref{eq:hyperbolicErrTarget1}, \eqref{eq:hyperbolicErrTarget2}, we show that $\tilde{w}$ is $0$ for $t \geq 1$.  
\end{proof}

We have now established that the observer in \eqref{eq:hyperbolicObs1}, \eqref{eq:hyperbolicObs2} converges to the true system exponentially. However, in practice, it can be challenging to implement the PDE form of the observer and so we provide the observer's explicit moving-horizon form in the next result. We assume that the horizon $T$ is not shorter than one, i.e., $T\geq 1$. We do so for the sake of simplifying the expressions of the MHE and for avoiding the need to store the past values of the estimate $\hat u(x,t-T)$.

\begin{theorem} \label{thm:hyperbolicExplicit}
 Let Assumption 1 hold. Then for every $U \in C^1(\mathbb{R}^+)$, $u_0, \hat{u}_0 \in C^1([0, 1])$ for which the compatibility conditions in \eqref{eq:hyperbolicCompat1}, \eqref{eq:hyperbolicCompat2},
    \eqref{eq:hyperbolicObsCompat1}, \eqref{eq:hyperbolicObsCompat2} hold, the solution of 
\eqref{eq:hyperbolicSys1}, \eqref{eq:hyperbolicSys2}, \eqref{eq:hyperbolicSys3}, \eqref{eq:hyperbolicObs1}, \eqref{eq:hyperbolicObs2} with initial conditions \eqref{iasson1}, \eqref{iasson2} satisfies for all $t \geq 1$
\begin{eqnarray}
    \hat{u}(x, t) = \hat{w}(x, t) - \int_0^x k(x, y) \hat{w}(y, t) dy \,, \label{eq:hyperbolicObsBckTransfrom}\, \hspace{5pt} \forall x \in [0, 1]\,,
\end{eqnarray}
where $\hat{w}$ is given by the formula
\begin{eqnarray}
    \hat{w}(x, t) &=& \int_{t+x-1}^t \theta(t+x-\tau) Y(\tau) d\tau \nonumber \\ && + U(t+x-1) \,, \hspace{5pt} \forall x \in [0, 1]\,,
    \label{eq:hyperbolicExplicit1}  
\end{eqnarray}
 with $\theta$  defined as
\begin{eqnarray}
\label{eq:hyperbolicExplicit2}  
    \theta(x) &=& p_1(x) + \int_0^x l(x, y) p_1(y) dy  \,,
\end{eqnarray}
where $p_1(x)$ is defined by \eqref{eq:hyperbolicObsGain}.
\end{theorem}

Note that Theorem \ref{thm:hyperbolicExplicit} allows us to use \eqref{eq:hyperbolicObsBckTransfrom} as an \textit{explicit moving-horizon estimator}. Furthermore, in \eqref{eq:hyperbolicExplicit2}, the functions $\theta(x)$ and $p_1(x)$  are in an inverse backstepping relationship, similar to the states of the $\hat{w}$ and $\hat{u}$ systems. 

\begin{proof}
    First, note that the observer system in \eqref{eq:hyperbolicObs1}, \eqref{eq:hyperbolicObs2} under the inverse backstepping transform
    \begin{eqnarray}
        \hat{w}(x, t) = \hat{u}(x, t) + \int_0^x l(x, y) \hat{u}(y, t) dy 
    \end{eqnarray}
    yields
    \begin{eqnarray}
    \label{simple-hyper1}
        \hat{w}_t(x, t) &=& \hat{w}_x(x, t) + \theta(x) u(0, t) \\ 
            \label{simple-hyper0}
        \hat{w}(1, t) &=& U(t) 
    \end{eqnarray}
    which is a transport PDE with a source term. The resulting PDE is solved explicitly as in \cite{BERNARD20142692} yielding the desired result. 
\end{proof}

Note that the moving-horizon estimator in 
\eqref{eq:hyperbolicObsBckTransfrom}, \eqref{eq:hyperbolicExplicit1}
 does not depend on past state estimates $\hat u(x,t-T)$, where $T>0$, at the start $t-T$ of the horizon, but just on the input and output signals for the PDE over the window $[t-1,t]$. This is due to the property of the transport PDE with a unity propagation speed that the effect of its initial condition "washes out" of the domain in the time interval $[0,1]$ and is zero for $t\geq 1$.

\section{1D Parabolic PDE} \label{sec:parabolic}
We now consider parabolic PDE systems of the form of
\begin{eqnarray}
    \label{eq:parabolicSys1} u_t(x, t) &=& u_{xx}(x, t) + \lambda(x) u(x, t), \hspace{5pt} \nonumber  \\ && x \in (0, 1)\,,  \\ 
    \label{eq:parabolicSys2} u(0, t) &=& 0 \,, \\ 
    \label{eq:parabolicSys3} u(1, t) &=& U(t) \,.
\end{eqnarray}
under the following assumption.

\begin{assumption} \label{assump2}
    $\lambda \in C^1([0, 1])$.
\end{assumption}

Moreover, in relation to Assumption \ref{assump2}, we define 
\begin{eqnarray}
    \bar{\lambda} = \sup_{x \in [0, 1]}|\lambda(x)|\,.
\end{eqnarray}

Assumption 2 guarantees that for every $u_0 \in H^2(0, 1)$ and every $U \in C^2(\mathbb{R}^+)$ with 
 \begin{eqnarray}
        u_0(0)&=&0  \label{iasson3}\,, \\
        u_0(1)&=&U(0) \label{iasson4}\,,
\end{eqnarray}
there exists a unique solution of the initial-boundary value problem \eqref{eq:parabolicSys1}, \eqref{eq:parabolicSys2}, \eqref{eq:parabolicSys3} with initial condition given by \eqref{iasson1}, which is defined for all $t \geq 0$ (see results for parabolic PDEs in the book \cite{evans10}). 

We study \eqref{eq:parabolicSys1}, \eqref{eq:parabolicSys2}, \eqref{eq:parabolicSys3} with (Neumann) boundary measurement 
    \begin{eqnarray}
        Y(t) = u_x(1, t) \,,
    \end{eqnarray}
    at $x=1$.

    Next we continue again by postulating an observer of the form 
    \begin{eqnarray}
        \hat{u}_t(x, t) &=& \hat{u}_{xx}(x, t) + \lambda(x) \hat{u}(x, t) \nonumber \\ && + p_1(x)\big[u_x(1, t) - \hat{u}_x(1, t)] \label{eq:parabolicObs1}\,, \\
        \hat{u}(0, t) &=& 0 \,,\label{eq:parabolicObs2} \\ 
        \hat{u}(1, t) &=& U(t) \label{eq:parabolicObs3}\,, 
    \end{eqnarray}
and with the observer gain function $p_1(x)$ to be derived. As above, we require that $\hat{u}_0 \in H^2(0, 1)$ and that the following compatibility conditions    
 \begin{eqnarray}
        \hat{u}_0(0)&=&0 \label{iasson6}\,, \\
        \hat{u}_0(1)&=&U(0) \label{iasson7}\,, 
    \end{eqnarray}
are satisfied to ensure the system \eqref{eq:parabolicSys1}, \eqref{eq:parabolicSys2}, \eqref{eq:parabolicSys3}, \eqref{eq:parabolicObs1}, \eqref{eq:parabolicObs2}, \eqref{eq:parabolicObs3} with initial conditions \eqref{iasson1} and \eqref{iasson2}
has a unique classical solution defined for all $t \geq 0$ (here a backstepping argument is first used in order to eliminate the nonlocal terms that appear on the right hand side of \eqref{eq:parabolicObs1}). 

    We need to choose $p_1(x)$ such that the error system is stable. As above, the error system can be written as
    \begin{eqnarray}
        \tilde{u}_t(x, t) &=& \tilde{u}_{xx}(x, t) + \lambda(x) \tilde{u}(x, t) \nonumber \\ && - p_1(x) \tilde{u}_x(1, t) \label{eq:parabolicError1}\,, \\ 
        \tilde{u}(0, t) &=& 0 \label{eq:parabolicError2}\,, \\
        \tilde{u}(1, t) &=& 0 \label{eq:parabolicError3}\,.
    \end{eqnarray}
    As the hyperbolic PDE system, our goal is to transform the system \eqref{eq:parabolicError1}, \eqref{eq:parabolicError2},
    \eqref{eq:parabolicError3} into the following target system
    \begin{eqnarray}
        \tilde{w}_t(x, t) &=& \tilde{w}_{xx}(x, t)\,,  \label{eq:parabolicWError1} \\ 
        \tilde{w}(0, t) &=& 0 \,, \label{eq:parabolicWError2}\\ 
        \tilde{w}(1, t) &=& 0 \,,\label{eq:parabolicWError3}
    \end{eqnarray}
    through the backstepping transform
    \begin{eqnarray}
        \tilde{u}(x, t) = \tilde{w}(x, t) - \int_x^1 k(x,y) \tilde{w}(y, t) dy. 
    \end{eqnarray}
    To achieve the transformation of \eqref{eq:parabolicError1}, \eqref{eq:parabolicError2},
    \eqref{eq:parabolicError3} into  
    \eqref{eq:parabolicWError1},
    \eqref{eq:parabolicWError2},
    \eqref{eq:parabolicWError3}, the kernel PDE $k(x, y)$ must satisfy the Goursat PDE
    \begin{eqnarray}
        \label{eq:parabolicKernelCond1} k_{xx}(x, y) - k_{yy}(x, y) &=& \lambda(y) k(x, y), \nonumber \hspace{5pt} \\ &&  \forall (x, y) \in \mathcal{T},\\ 
        \label{eq:parabolicKernelCond2}k(x, 0) &=& 0, \hspace{10pt} \forall x \in [0, 1] \,,\\\ 
       \label{eq:parabolicKernelCond3} k(x, x) &=& -\frac{1}{2} \int_0^x \lambda(y) dy, \hspace{5pt} \nonumber \\ && \forall x \in [0, 1] \,,
    \end{eqnarray}
    where $\mathcal{T} = \{(x, y) \in \mathbb{R}^2 ; 0 \leq y \leq x \leq 1 \}$ is the same as above. Furthermore,  the observer gain $p_1(x)$ must be
    \begin{eqnarray}
        p_1(x) = k(1, x)\,.
    \end{eqnarray}
    As above, the resulting kernel PDE $k(x, y)$ is bounded. 

\begin{theorem}
    (proven in \cite{1369395}, \cite{Smyshlyaev2010}) For every $\lambda \in C^1([0, 1])$, the PDE problem   \eqref{eq:parabolicKernelCond1}, \eqref{eq:parabolicKernelCond2}, \eqref{eq:parabolicKernelCond3} has a unique $\text{C}^2(\mathcal{T})$ solution with the bound
    \begin{eqnarray}
        |k(x, y)| \leq \bar{\lambda}e^{2\bar{\lambda}x} \,,\label{eq:parabolicKernelBound}
    \end{eqnarray}
    for all $x \in [0, 1]$.
\end{theorem}

Additionally, recall the inverse backstepping kernel $l(x, y) \in C^2(\mathcal{T})$ and inverse backstepping transformation 
\begin{eqnarray}
    \tilde{w}(x, t) &=& \tilde{u}(x, t) + \int_x^1 l(x, y) \tilde{u}(y, t) dy\,,
\end{eqnarray}
which similarly satisfies \eqref{eq:hyperbolicKernelProp}. Then, the inverse kernel $l(x, y)$ is bounded as
\begin{eqnarray}
     |l(x, y)| \leq (\bar{\lambda} e^{2 \bar{\lambda}}) e^{\bar{\lambda} e^{2 \bar{\lambda}}} \label{eq:parabolicInvKernelBound} \,.
\end{eqnarray}
Now, we provide the error system convergence result for the observer in \eqref{eq:parabolicObs1}, \eqref{eq:parabolicObs2}, 
\eqref{eq:parabolicObs3}.

\begin{theorem} (proven in \cite{krstic2023neural})
    Let Assumption \ref{assump2} hold. Then for every $U \in C^2(\mathbb{R}^+)$, $u_0, \hat{u}_0 \in H^2(0, 1)$ for which the compatibility conditions \eqref{iasson3}, \eqref{iasson4}, \eqref{iasson6}, \eqref{iasson7} hold, the observer \eqref{eq:parabolicObs1}, \eqref{eq:parabolicObs2}, \eqref{eq:parabolicObs3} ensures that the error between its state with initial condition \eqref{iasson2} and the state of the plant  \eqref{eq:parabolicSys1}, \eqref{eq:parabolicSys2}, \eqref{eq:parabolicSys3} with initial condition \eqref{iasson1} satisfies for all $t \geq 0$ the exponential stability bound
\begin{eqnarray}
    \|u[t]-\hat{u}[t]\| \leq Me^{-\pi^2 t}\|u_0 - \hat{u}_0\|, \hspace{5pt} \forall t \geq 0 \,,
\end{eqnarray}
where 
\begin{eqnarray}\label{eq-overshootM}
    M = \left(1+\bar{\lambda}e^{2 \bar{\lambda}} \right) \left(1 + \bar{\lambda}e^{2\bar{\lambda}} \right) e^{\bar{\lambda}e^{2 \bar{\lambda}}}.
\end{eqnarray}
Furthermore, over a horizon of length $T$, the estimation error obeys the relation for all $t \geq T$
\begin{eqnarray}
    \|u[t] - \hat{u}[t]\| &\leq& Me^{-\pi^2 T}\|u[t-T] - \hat{u}[t-T]\| \,.
    \label{eq:parabolicContract}
\end{eqnarray}
\end{theorem}

Note, that the relation \eqref{eq:parabolicContract}, which quantifies the progress in the state estimation over the horizon $[t-T,t]$, becomes a {\em contraction} when 
\begin{eqnarray}
    T > \frac{\ln(M)}{\pi^2}\,,
\end{eqnarray} 
for the overshoot coefficient $M$ given by \eqref{eq-overshootM} and dependent on the degree of plant's instability $\bar\lambda$, i.e., when the horizon is long enough in relation to the plant's level of instability. 

The next result is arguably the paper's main result: the explicit MHE for the reaction-diffusion PDE. We consider the following {\em explicit moving-horizon estimator} 
    \begin{align}
        \label{eq:parabolicUhatBckstep}
        \hat{u}(x, t) =& \hat{w}(x, t) - \int_x^1 k(y, x) 
          \hat{w}(y, t) dy 
\\  
    \label{eq:whatsolution}
    \hat w(x,t) 
    =& \sum_{n=1}^\infty \phi_n(x) \exp(-n^2 \pi^2 T) \nonumber \\ & \times \int_0^1 \hat{w}(y, t-T) \phi_n(y) d y \nonumber
    \\ & + \sum_{n=1}^\infty \left( \int_0^1 l(1, y) \phi_n(y) dy \right) \phi_n(x) \nonumber \\
    & \times \int_{t-T}^t \exp{\big(-n^2 \pi^2 (t-\tau) \nonumber \big)}\\ & \times \nonumber \big(Y(\tau)-k(1, 1)U(\tau)\big) d \tau
    \\ & -\pi \sqrt{2} \sum_{n=1}^\infty n \big(-1\big)^n \phi_n(x) \nonumber \\ & \times \int_{t-T}^t \exp{\big(-n^2 \pi^2 (t-\tau) \big)}U(\tau) d\tau
\\
\label{eq:whatinit}
            \hat w(x,t-T) =& \hat u(x,t-T) 
            +\int_x^1 l(y,x) \hat u(y,t-T) dy 
\end{align}
where $\phi_n(x) = \sqrt{2}\sin(n\pi x) $ is the $n$th eigenfunction of the underlying Sturm-Liouville operator for the heat equation with homogeneous Dirichlet boundary conditions. The moving-horizon estimator is well defined for $t \geq T$ and its properties are given by the following theorem. 
\begin{theorem}
Let Assumption 2 hold. Then for every $U \in C^2(\mathbb{R}^+)$, $u_0, \hat{u}_0 \in H^2(0, 1)$ for which the compatibility conditions in \eqref{iasson3}, \eqref{iasson4}, \eqref{iasson6}, \eqref{iasson7} hold, the solution of \eqref{eq:parabolicSys1}, \eqref{eq:parabolicSys2}, 
\eqref{eq:parabolicSys3},
\eqref{eq:parabolicObs1}, 
\eqref{eq:parabolicObs2}
with initial conditions \eqref{iasson1}, \eqref{iasson2} satisfies for all $t\geq T$ \eqref{eq:parabolicUhatBckstep} where $\hat{w}$ is given by \eqref{eq:whatsolution} and \eqref{eq:whatinit}.
\end{theorem}

\begin{proof}
    Note, from \cite{iasson}, the explicit form of  \eqref{eq:whatsolution}, \eqref{eq:whatinit} is equivalent to the  system
    \begin{eqnarray}
        \label{eq:parabolicTarget1} \hat{w}_t &=& \hat{w}_{xx} + l(1, x) \bigg[ u_x(1, t) - k(1, 1)U(t) \bigg] \\ 
        \label{eq:parabolicTarget2} \hat{w}(0, t) &=& 0 \\ 
        \label{eq:parabolicTarget3} \hat{w}(1, t) &=& U(t) 
    \end{eqnarray}
    on the interval $[t-T, t]$. 
    The system \eqref{eq:parabolicTarget1}, \eqref{eq:parabolicTarget2}, \eqref{eq:parabolicTarget3} is transformed into the following observer system after the application of the inverse backstepping transformation as 
    \begin{eqnarray}
        \hat{w}(x, t) = \hat{u}(x, t) + \int_x^1 l(y, x) \hat{u}(y, t) dy 
    \end{eqnarray}
    resulting in the observer \eqref{eq:parabolicObs1}, \eqref{eq:parabolicObs2}, \eqref{eq:parabolicObs3}
   and obtaining the asserted result.
\end{proof}

\begin{figure}[ht]
    \centering
    \includegraphics{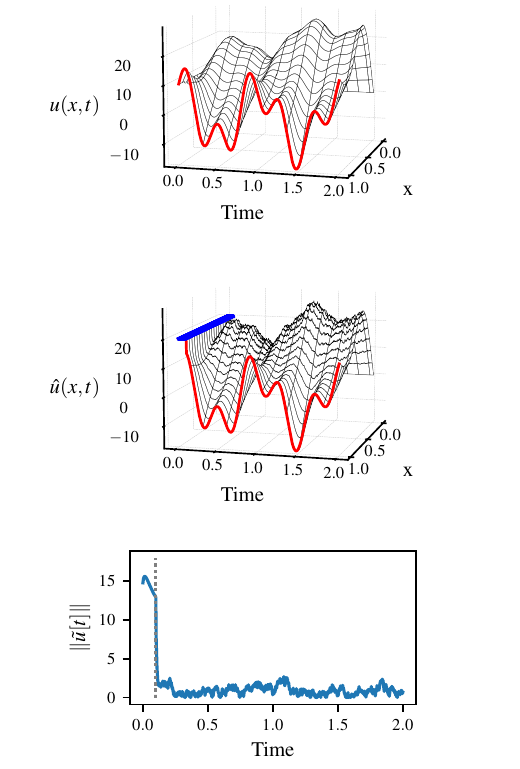}
    \vspace{-2em}
    \caption{The top row shows system \eqref{eq:parabolicSys1}, \eqref{eq:parabolicSys2}, \eqref{eq:parabolicSys3} with $u_0(x) \equiv 10$ and input $U(t) = 10\cos{(2 \pi t)} + 7 \sin{(16t)}$, which creates an unstable and input-excited system. The second row represents the observer in the explicit MHE form \eqref{eq:parabolicUhatBckstep}, \eqref{eq:whatsolution}, \eqref{eq:whatinit}, with the sliding window as $T=0.1$ and the infinite series in \eqref{eq:whatsolution} truncated to $N=4$ terms. We add Gaussian measurement noise on the boundary measurement in the form of $\text{Normal}(0, 500)$. The large noise is handled due the frequency of measurements available when using a time step of $1e-6$. In practice, one may not be able to handle such a large deviation with a less rapid measurement frequency. The initial condition for the observer for $t \leq T$ is $\hat{u}(x, \theta) = 20,$ $\theta \in [0, T]$. The bottom row showcases the $L^2$ error between the observer system and the true system, $\|\tilde u[t]\|_{L^2}$, which decreases precipitously at $t=0.1$ (shown by the grey dotted line), when the observer estimate begins to engage.}
    \label{fig:parabolicSimulations}
    \vspace{-1.5em}
\end{figure}

\section{Simulation of Parabolic MHE} \label{sec:simParabolic}
 We consider the system in \eqref{eq:parabolicSys1}, \eqref{eq:parabolicSys2}, \eqref{eq:parabolicSys3} with $\lambda(x) = 21 \cos{(5 \cos^{-1}(x))}$ as a Chebyshev polynomial. Additionally, we have an input controller $U(t) = 10\cos(2\pi t) + 7 \sin(16 t)$ which creates a slowly unstable system as shown in the top row of Figure \ref{fig:parabolicSimulations}. For computation, we choose a horizon window of $T = 0.1$ over $2$ total seconds of simulation time and set $dt =1e-6$ and $dx = 0.001$ as the temporal and spatial step sizes respectively. To simulate the true PDE system. we employ a traditional forward Euler finite difference scheme as in \cite{krstic2023neural}. We truncate the infinite series in \eqref{eq:whatsolution} $4$ terms which can be precomputed offline. The \textit{only} online computation are the discrete integrals in \eqref{eq:whatsolution} which depend on the horizon interval $T$ and can be computed recursively from the previous timestep. 

    In the second row of Figure \ref{fig:parabolicSimulations}, we chose for the  observer to have an initial condition of $\hat{u}(x, \theta)=20, \theta \in [0, T]$ for the initial horizon window and add Gaussian noise on the measurement of $u_x(1, t)$. Despite the noise,  the bottom of Figure \ref{fig:parabolicSimulations} shows the $L_2$ error quickly decreases once the estimate engages at $t=0.1$ and that the observer is able to emulate the true state with good accuracy. We note, without noise, the observer would exponentially converge to the true state immediately after engaging.

\section{Conclusion} \label{sec:conclusion}
In this work, we present the first set of \textit{moving horizon estimators} for both first-order hyperbolic and reaction-diffusion PDEs. The MHEs replace the need to employ burdensome numerical schemes for solving PDE observers with explicit forms. It is through the power of a backstepping transformation that establishes the equivalence between a complicated unsolvable PDE and a simple, explicitly solvable (target) PDE. 
To the authors' knowledge, this is the first work in designing, for PDE systems, moving-horizon estimators, which are well developed for discrete-time models, in the context of MPC \cite{rawlings:mayne:diehl:2020,Rawlings2003}. 
It is still unclear what the full potential of the PDE backstepping-based MHEs is, but one suggestion, along the lines of \cite{krstic2023neural, bhan2023neural, pmlr-v211-bhan23a, 9992759}, is to approximate the MHE operators with neural networks. The finite-horizon nature of the MHEs makes them suitable for a deep learning representation in which the input functions (the measured output and control input signals, as well as the state estimate from the start of the horizon) are functions on constant intervals whose lengths are constant (and, in particular, do not grow unbounded). The authors leave this topic as an exciting possibility for future work.

\bibliography{references}

\begin{thebibliography}{10}

\bibitem{10.1007/978-1-4471-0853-5_19}
F.~Allg{\"o}wer, T.~A. Badgwell, J.~S. Qin, J.~B. Rawlings, and S.~J. Wright.
\newblock Nonlinear predictive control and moving horizon estimation --- an introductory overview.
\newblock In P.~M. Frank, editor, {\em Advances in Control}, pages 391--449, London, 1999. Springer London.

\bibitem{BERNARD20142692}
P.~Bernard and M.~Krstic.
\newblock Adaptive output-feedback stabilization of non-local hyperbolic {PDE}s.
\newblock {\em Automatica}, 50(10):2692--2699, 2014.

\bibitem{bhan2023neural}
L.~Bhan, Y.~Shi, and M.~Krstic.
\newblock Neural operators for bypassing gain and control computations in {PDE} backstepping.
\newblock {\em IEEE Transactions on Automatic Control}, pages 1--16, 2023.
\newblock \url{https://ieeexplore.ieee.org/document/10374221}.

\bibitem{pmlr-v211-bhan23a}
L.~Bhan, Y.~Shi, and M.~Krstic.
\newblock Operator learning for nonlinear adaptive control.
\newblock In N.~Matni, M.~Morari, and G.~J. Pappas, editors, {\em Proceedings of The 5th Annual Learning for Dynamics and Control Conference}, volume 211 of {\em Proceedings of Machine Learning Research}, pages 346--357. PMLR, 15--16 Jun 2023.

\bibitem{6573344}
F.~Di~Meglio, R.~Vazquez, and M.~Krstic.
\newblock Stabilization of a system of $n+1$ coupled first-order hyperbolic linear {PDE}s with a single boundary input.
\newblock {\em IEEE Transactions on Automatic Control}, 58(12):3097--3111, 2013.

\bibitem{DONGMO202285}
M.~K.~J. Dongmo and T.~Meurer.
\newblock Moving horizon estimator design for a nonlinear diffusion-reaction system with sensor dynamics.
\newblock {\em IFAC-PapersOnLine}, 55(20):85--90, 2022.
\newblock 10th Vienna International Conference on Mathematical Modelling MATHMOD 2022.

\bibitem{evans10}
L.~C. Evans.
\newblock {\em Partial differential equations}.
\newblock American Mathematical Society, Providence, R.I., 2010.

\bibitem{1039801}
G.~Ferrari-Trecate, D.~Mignone, and M.~Morari.
\newblock Moving horizon estimation for hybrid systems.
\newblock {\em IEEE Transactions on Automatic Control}, 47(10):1663--1676, 2002.

\bibitem{doi:10.1021/ie034308l}
E.~L. Haseltine and J.~B. Rawlings.
\newblock Critical evaluation of extended kalman filtering and moving-horizon estimation.
\newblock {\em Industrial \& Engineering Chemistry Research}, 44(8):2451--2460, 2005.

\bibitem{JANG2014159}
H.~Jang, J.~H. Lee, R.~D. Braatz, and K.-K.~K. Kim.
\newblock Fast moving horizon estimation for a two-dimensional distributed parameter system.
\newblock {\em Computers \& Chemical Engineering}, 63:159--172, 2014.

\bibitem{iasson}
I.~Karafyllis and M.~Krstic.
\newblock {\em Input-to-State Stability for {PDE}s}.
\newblock Springer Cham, 1 edition, 2019.

\bibitem{KRAUS201325}
T.~Kraus, H.~Ferreau, E.~Kayacan, H.~Ramon, J.~{De Baerdemaeker}, M.~Diehl, and W.~Saeys.
\newblock Moving horizon estimation and nonlinear model predictive control for autonomous agricultural vehicles.
\newblock {\em Computers and Electronics in Agriculture}, 98:25--33, 2013.

\bibitem{krstic2023neural}
M.~Krstic, L.~Bhan, and Y.~Shi.
\newblock Neural operators of backstepping controller and observer gain functions for reaction-diffusion {PDE}s.
\newblock {\em arXiv}, 2023.
\newblock \url{https://arxiv.org/abs/2303.10506}.

\bibitem{10.1115/1.3023128}
M.~Krstic, L.~Magnis, and R.~Vazquez.
\newblock {Nonlinear Control of the Viscous Burgers Equation: Trajectory Generation, Tracking, and Observer Design}.
\newblock {\em Journal of Dynamic Systems, Measurement, and Control}, 131(2):021012, 02 2009.

\bibitem{krstic2008boundary}
M.~Krstic and A.~Smyshlyaev.
\newblock {\em Boundary Control of {PDE}s: A Course on Backstepping Designs}.
\newblock SIAM, 2008.

\bibitem{KUNTZ2023108429}
S.~J. Kuntz, J.~J. Downs, S.~M. Miller, and J.~B. Rawlings.
\newblock An industrial case study on the combined identification and offset-free control of a chemical process.
\newblock {\em Computers \& Chemical Engineering}, 179:108429, 2023.

\bibitem{7572119}
A.~Liu, W.-A. Zhang, M.~Z.~Q. Chen, and L.~Yu.
\newblock Moving horizon estimation for mobile robots with multirate sampling.
\newblock {\em IEEE Transactions on Industrial Electronics}, 64(2):1457--1467, 2017.

\bibitem{LIU201661}
B.-N. Liu, D.~Boutat, and D.-Y. Liu.
\newblock Backstepping observer-based output feedback control for a class of coupled parabolic {PDE}s with different diffusions.
\newblock {\em Systems \& Control Letters}, 97:61--69, 2016.

\bibitem{socMoura}
S.~Moura, N.~Chaturvedi, and M.~Krstic.
\newblock Adaptive {PDE} observer for battery {SOC}/{SOH} estimation via an electrochemical model.
\newblock {\em ASME Journal of Dynamic Systems, Measurement, and Control}, 136:011015--011026, 01 2014.

\bibitem{Muske1995}
K.~R. Muske and J.~B. Rawlings.
\newblock {\em Nonlinear Moving Horizon State Estimation}, pages 349--365.
\newblock Springer Netherlands, Dordrecht, 1995.

\bibitem{4792993}
K.~R. Muske, J.~B. Rawlings, and J.~H. Lee.
\newblock Receding horizon recursive state estimation.
\newblock In {\em 1993 American Control Conference}, pages 900--904, 1993.

\bibitem{Rawlings2003}
C.~Rao, J.~Rawlings, and D.~Mayne.
\newblock Constrained state estimation for nonlinear discrete-time systems: stability and moving horizon approximations.
\newblock {\em IEEE Transactions on Automatic Control}, 48(2):246--258, 2003.

\bibitem{RAO20011619}
C.~V. Rao, J.~B. Rawlings, and J.~H. Lee.
\newblock Constrained linear state estimation—a moving horizon approach.
\newblock {\em Automatica}, 37(10):1619--1628, 2001.

\bibitem{Rawlings20061529}
J.~B. Rawlings and B.~R. Bakshi.
\newblock Particle filtering and moving horizon estimation.
\newblock {\em Computers \& Chemical Engineering}, 30(10):1529--1541, 2006.
\newblock Papers form Chemical Process Control VII.

\bibitem{rawlings:mayne:diehl:2020}
J.~B. Rawlings, D.~Q. Mayne, and M.~M. Diehl.
\newblock {\em Model Predictive Control: Theory, Design, and Computation}.
\newblock Nob Hill Publishing, Santa Barbara, CA, 2nd, paperback edition, 2020.
\newblock 770 pages, ISBN 978-0-9759377-5-4.

\bibitem{https://doi.org/10.1002/aic.690420811}
D.~G. Robertson, J.~H. Lee, and J.~B. Rawlings.
\newblock A moving horizon-based approach for least-squares estimation.
\newblock {\em AIChE Journal}, 42(8):2209--2224, 1996.

\bibitem{9992759}
Y.~Shi, Z.~Li, H.~Yu, D.~Steeves, A.~Anandkumar, and M.~Krstic.
\newblock Machine learning accelerated {PDE} backstepping observers.
\newblock In {\em 2022 IEEE 61st Conference on Decision and Control (CDC)}, pages 5423--5428, 2022.

\bibitem{1369395}
A.~Smyshlyaev and M.~Krstic.
\newblock Closed-form boundary state feedbacks for a class of 1-{D} partial integro-differential equations.
\newblock {\em IEEE Transactions on Automatic Control}, 49(12):2185--2202, 2004.

\bibitem{SMYSHLYAEV2005613}
A.~Smyshlyaev and M.~Krstic.
\newblock Backstepping observers for a class of parabolic {PDE}s.
\newblock {\em Systems \& Control Letters}, 54(7):613--625, 2005.

\bibitem{Smyshlyaev2010}
A.~Smyshlyaev and M.~Krstic.
\newblock {\em Adaptive Control of Parabolic {PDE}s}.
\newblock Princeton University Press, 2010.

\bibitem{STEEVES20203}
D.~Steeves, M.~Krstic, and R.~Vazquez.
\newblock Prescribed–time estimation and output regulation of the linearized {S}chrödinger equation by backstepping.
\newblock {\em European Journal of Control}, 55:3--13, 2020.
\newblock Finite-time estimation, diagnosis and synchronization of uncertain systems.

\bibitem{1583115}
R.~Vazquez and M.~Krstic.
\newblock A closed-form observer for the channel flow {N}avier-{S}tokes system.
\newblock In {\em Proceedings of the 44th IEEE Conference on Decision and Control}, pages 5959--5964, 2005.

\bibitem{5208259}
R.~Vazquez and M.~Krstic.
\newblock Boundary observer for output-feedback stabilization of thermal-fluid convection loop.
\newblock {\em IEEE Transactions on Control Systems Technology}, 18(4):789--797, 2010.

\bibitem{6160338}
R.~Vazquez, M.~Krstic, and J.-M. Coron.
\newblock Backstepping boundary stabilization and state estimation of a 2 × 2 linear hyperbolic system.
\newblock In {\em 2011 50th IEEE Conference on Decision and Control and European Control Conference}, pages 4937--4942, 2011.

\bibitem{4177273}
R.~Vazquez, E.~Schuster, and M.~Krstic.
\newblock A closed-form observer for the 3{D} inductionless {MHD} and {N}avier-{S}tokes channel flow.
\newblock In {\em Proceedings of the 45th IEEE Conference on Decision and Control}, pages 739--746, 2006.

\bibitem{VAZQUEZ20082517}
R.~Vazquez, E.~Schuster, and M.~Krstic.
\newblock Magnetohydrodynamic state estimation with boundary sensors.
\newblock {\em Automatica}, 44(10):2517--2527, 2008.

\bibitem{WANG20141581}
S.~Wang, L.~Chen, D.~Gu, and H.~Hu.
\newblock An optimization based moving horizon estimation with application to localization of autonomous underwater vehicles.
\newblock {\em Robotics and Autonomous Systems}, 62(10):1581--1596, 2014.

\bibitem{8685160}
L.~Zou, Z.~Wang, Q.-L. Han, and D.~Zhou.
\newblock Moving horizon estimation for networked time-delay systems under round-robin protocol.
\newblock {\em IEEE Transactions on Automatic Control}, 64(12):5191--5198, 2019.

\end{thebibliography}
\bibliographystyle{abbrv}

\end{document}